\documentclass[conference]{IEEEtran}

\usepackage{xcolor,soul,framed} 

\colorlet{shadecolor}{yellow}
\usepackage[pdftex]{graphicx}
\graphicspath{{../pdf/}{../jpeg/}}
\DeclareGraphicsExtensions{.pdf,.jpeg,.png}
\usepackage{amsmath,amssymb,amsfonts}
\usepackage{array}
\usepackage{mdwmath}
\usepackage{mdwtab}
\usepackage{eqparbox}
\usepackage{url}
\usepackage{caption}
\usepackage{xcolor}
\usepackage{soul}
\usepackage{verbatim}
\usepackage{float}
\usepackage[font=scriptsize]{subcaption}
\usepackage[font=scriptsize]{caption}
\usepackage{amsmath,amssymb,amsfonts}
\usepackage{algorithm}
\usepackage{wrapfig}
\usepackage{float}
\usepackage[noend]{algpseudocode}
\usepackage{soul}

\usepackage{multirow}

\newtheorem{Claim}{Claim}

\DeclareFontFamily{U}{mathx}{\hyphenchar\font45}
\DeclareFontShape{U}{mathx}{m}{n}{
      <5> <6> <7> <8> <9> <10>
      <10.95> <12> <14.4> <17.28> <20.74> <24.88>
      mathx10
      }{}
\DeclareSymbolFont{mathx}{U}{mathx}{m}{n}
\DeclareMathSymbol{\bigtimes}{1}{mathx}{"91}
\begin{document}
\title{Learning-Enabled Adaptive Voltage Protection Against Load Alteration Attacks On Smart Grids}
\author{
  \IEEEauthorblockN{Anjana B.\IEEEauthorrefmark{1}, Suman Maiti\IEEEauthorrefmark{1}, Sunandan Adhikary\IEEEauthorrefmark{1}, Soumyajit Dey\IEEEauthorrefmark{1}, Ashish R. Hota\IEEEauthorrefmark{2}}
  \IEEEauthorblockA{\IEEEauthorrefmark{1}Department of Computer Science, Indian Institute of Technology  Kharagpur\\ Email: anjanab@kgpian.iitkgp.ac.in, sumanmaiti99@kgpian.iitkgp.ac.in, mesunandan@kgpian.iitkgp.ac.in, soumya@cse.iitkgp.ac.in}
  \IEEEauthorblockA{\IEEEauthorrefmark{2}Department of Electrical Engineering,  Indian Institute of Technology  Kharagpur\\ Email: ahota@ee.iitkgp.ac.in}
}
\maketitle

\begin{abstract}
Smart grids are designed to efficiently handle variable power demands, especially for large loads, by real-time monitoring, distributed generation and distribution of electricity. However, the grid's distributed nature and the internet connectivity of large loads like Heating Ventilation, and Air Conditioning (HVAC) systems introduce vulnerabilities in the system that cyber-attackers can exploit, potentially leading to grid instability and blackouts. Traditional protection strategies,  primarily designed to handle transmission line faults are often inadequate against such threats, emphasising the need for enhanced grid security. In this work, we propose a Deep Reinforcement Learning (DRL)-based protection system that learns to differentiate any stealthy load alterations from normal grid operations and adaptively adjusts activation thresholds of the protection schemes. We train this adaptive protection scheme against an optimal and stealthy load alteration attack model that manipulates the power demands of HVACs at the most unstable grid buses to induce blackouts. We theoretically prove that the adaptive protection system trained in this competitive game setting can effectively mitigate any stealthy load alteration-based attack. To corroborate this, we also demonstrate the method's success in several real-world grid scenarios by implementing it in a hardware-in-loop setup.
\end{abstract}
%
%
\IEEEpeerreviewmaketitle
\section{Introduction}
\label{secIntro}


A significant portion of loads in modern smart grids are cyber-enabled Internet of Things (IoT)  based devices. Being resource-constrained, such devices are vulnerable to cyber attacks that can manipulate device operations. For example, consider attacks like \cite{soltan2018blackiot, yang2023resilient, shekari2022madiot, huang2019not} which resulted in strategic  power flow fluctuations in grid systems. 
 In today's world, Heating, Ventilation, and Air Conditioning (HVAC) systems  consume a significant amount of energy; accounting for $\sim$20\% of overall grid energy usage.
 Modern HVACs are often operated as cyber-enabled, i.e. they act as high-energy IoT loads.  A classic example of a cyber-attack targeting HVAC systems is reported in \cite{HVAC_access}, where attackers compromised HVAC operations via remote access to its web service, ultimately causing grid blackouts.

\par To prevent demand-supply imbalances in the grid caused by natural events, such as transmission line faults due to short circuits, smart grids employ protection mechanisms. These mechanisms shed loads \cite{huang2019not} and trip transmission lines and generators \cite{soltan2018blackiot} using relays and circuit breakers to prevent blackout conditions.  However, these mechanisms, primarily designed to handle natural faults, are ineffective against the previously discussed strategically crafted grid attacks \cite{HVAC_access}. The ineffectiveness of existing protection mechanisms can be attributed to two main reasons. \textit{(i)} Protection schemes activate 
after a constant delay from the occurrence of a power imbalance
due to their predefined current and voltage thresholds.
\textit{(ii)} They are not designed with cyber or physical level attacks in mind.
\par \textbf{Related Work:} There exist several research works that aim to design mitigation strategies for cyber-physical attacks by designing vulnerability modeling and defense simulation frameworks~\cite{koley2021catch,tan2016optimal}. 
The works in~\cite{ni2019multistage, cunningham2022deep} propose a {zero-sum security game} between the attacker and defender in power grids. 
While such analysis is standard for ensuring stability in power grids, these attack approaches are not practical.
%
It has been shown in \cite{shekari2022madiot, soltan2018blackiot} that LAAs are difficult to detect when launched by manipulating the load consumption of IoT devices by existing bad-data detectors.
%
The mitigation strategy in \cite{shekari2022madiot} lacked a theoretical foundation as it adjusted the activation delays of grid protection schemes by specific amounts based on a survey of grid operators. The works in \cite{umsonst2020nash, koley2021catch} propose defence techniques against advanced FDIAs that update their strategies to maintain stealth
while maintaining a minimum false positive rate (FPR). 
However, these are not directly applicable to complex distributed systems like smart grids due to the dynamic nature of grid operations, where defence mechanisms must account for non-linear power flow equations and time-varying power demand.
\par \textbf{Motivation and Problem Formulation:}
Consider a standard IEEE 14-bus power grid model with dynamic HVAC loads. At the 2-minute mark of the simulation of this model, the power consumption is at its peak. Say, we falsify the operating temperature setpoint of the HVAC load connected to bus number $3$ starting from this time, which increases its power consumption by $4$ Mega Watt (MW). Fig.~\ref{MotFigFreq} illustrates the impact of this load alteration on the grid's operating frequency in the time (x-axis) vs grid frequency (y-axis) plot. As can be seen from the figure, the frequency response due to the load alteration (blue plot) oscillates but remains within the upper (red plot) and lower (black plot) safety limits \cite{kundur2007power}. This occurs because the load alteration is introduced during the peak power consumption period. As a result, the automatic generation control (AGC) program in the grid anticipates the increased power demand as part of normal grid conditions and meets this higher demand. 
\begin{figure}[h]
\centering

\begin{subfigure}[b]{0.48\columnwidth}
   \includegraphics[clip,width=\columnwidth]{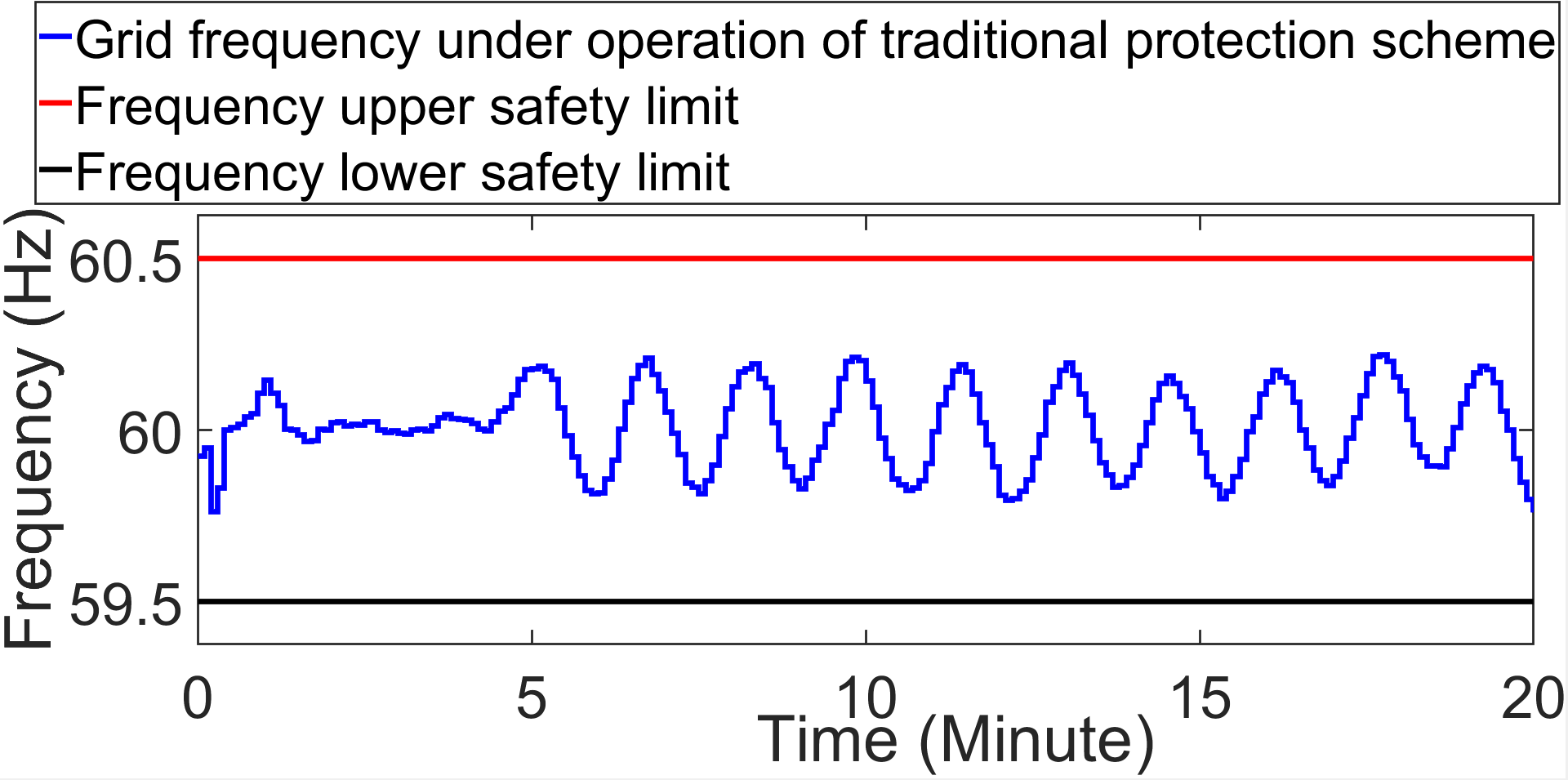}
   \caption{Grid operating frequency}
   \label{MotFigFreq}
\end{subfigure}
\hspace{0.01\columnwidth} 
\begin{subfigure}[b]{0.48\columnwidth}
    \includegraphics[clip,width=\columnwidth]{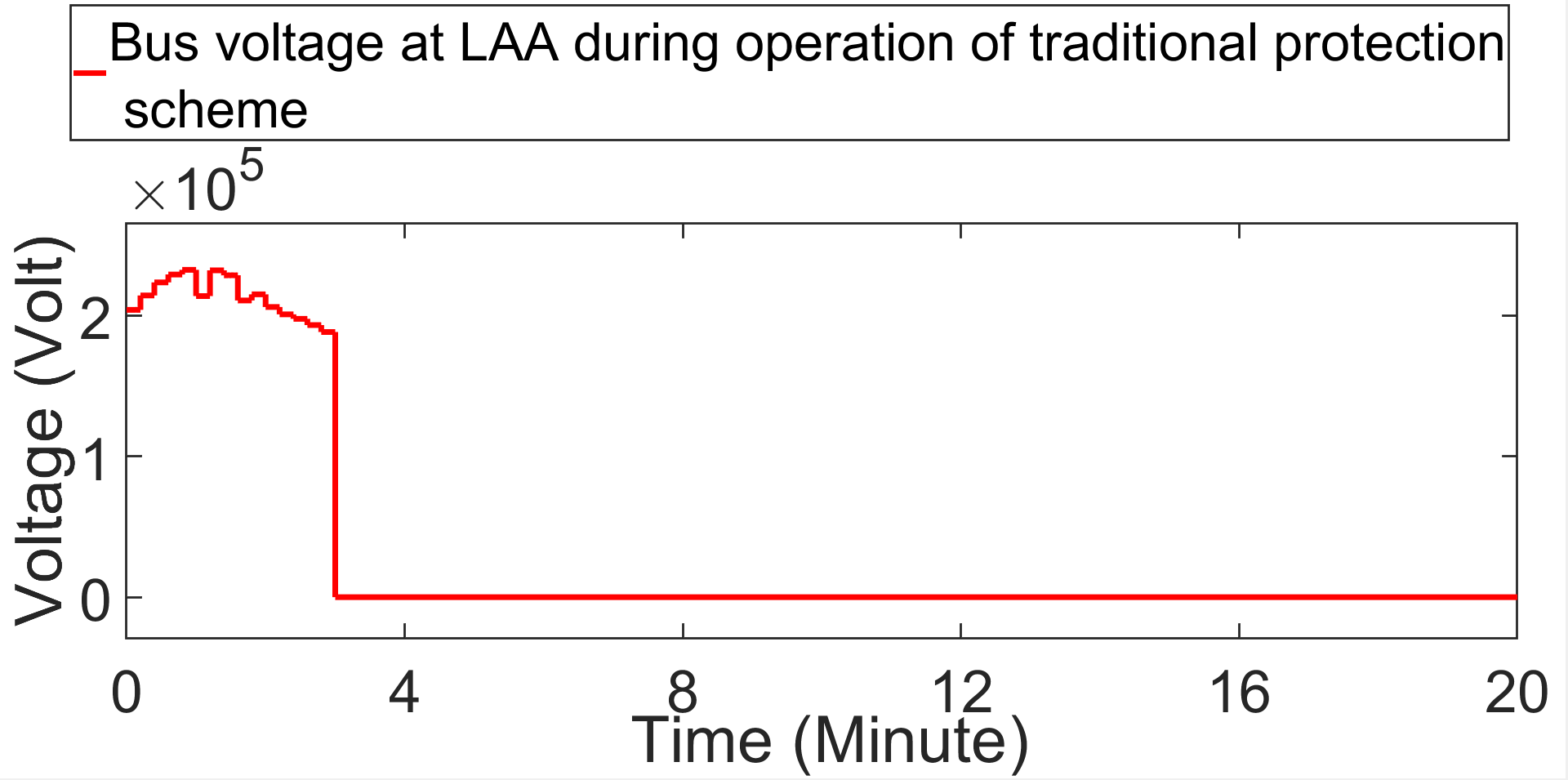}
    \caption{Phase a voltage waveform of bus 3}
    \label{MotFigVolt}
\end{subfigure}
\caption{Grid operating frequency and voltage waveform}
\vspace{-2mm}
\label{fig:figMotEx}
\end{figure}
But as shown in the time vs bus $3$ voltage plot in Fig.~\ref{MotFigVolt}, the phase \emph{a} voltage (red plot) of bus 3 significantly diminishes from the onset of the attack and drops to zero at approximately $2.1$ minutes, indicating a blackout condition. This occurs because traditional voltage protection schemes fail to detect the rapid Rate of Change of Voltage (RoCoV) caused by such a {\em stealthy load alteration} (SLA) via sensor data falsification. This raises a few research questions (\textbf{RQ}) regarding modern power grid security. \textbf{RQ1}: Is it possible to differentiate between peak power demand during normal grid operation and SLA-induced power demand by observing the RoCoV at grid buses? \textbf{RQ2}: Can we design an adaptive thresholding scheme for the existing voltage protection systems to mitigate any potential SLA scenarios? 
%
\par\textbf{Novelty and Contributions:}
To address these questions, in this work, we propose a comprehensive methodology for designing a novel Adaptive Protection System (APS) that adaptively updates the activation thresholds for existing voltage protection schemes to mitigate different LAA scenarios. Fig.~\ref{fig:system_model} summarises the overall proposed framework that designs an APS policy to mitigate any SLA with a game-theoretic guarantee.
We summarize the contributions as follows.
\par\noindent\textbf{1.} We develop a deep reinforcement learning (DRL)-based Adaptive Protection System (APS) that learns to adaptively adjust the activation thresholds of existing voltage protection schemes (annotated with 2 in Fig.~\ref{fig:system_model}) by sensing the presence of stealthy load alterations (SLAs). This  agent (green outlined box in Fig.~\ref{fig:system_model}) is trained in a competitive Multi-Agent RL (MARL) environment in the presence of an adversarial agent that launches such SLAs (red outlined box in Fig.~\ref{fig:system_model}).
\par\noindent\textbf{2.} To model an adversary as a load alteration attacker (LAA) with model-specific knowledge, an optimal SLA algorithm is designed that falsifies the temperature sensor data of IoT HVAC loads connected to the most unstable bus (annotated with 1 in Fig.~\ref{fig:system_model}).
By proving the existence of an equilibrium in this {\em two-player zero-sum attacker-mitigator game}, we theoretically ensure that the APS policy, output by our framework, mitigates any LAA strategy.
\par\noindent\textbf{3.} We validate the real-time applicability of our attack mitigation framework in a hardware-in-loop (HIL)  implementation. 
%
\begin{figure}[h!]
    \centering
    \vspace{-3mm}
    \includegraphics[width=\linewidth]{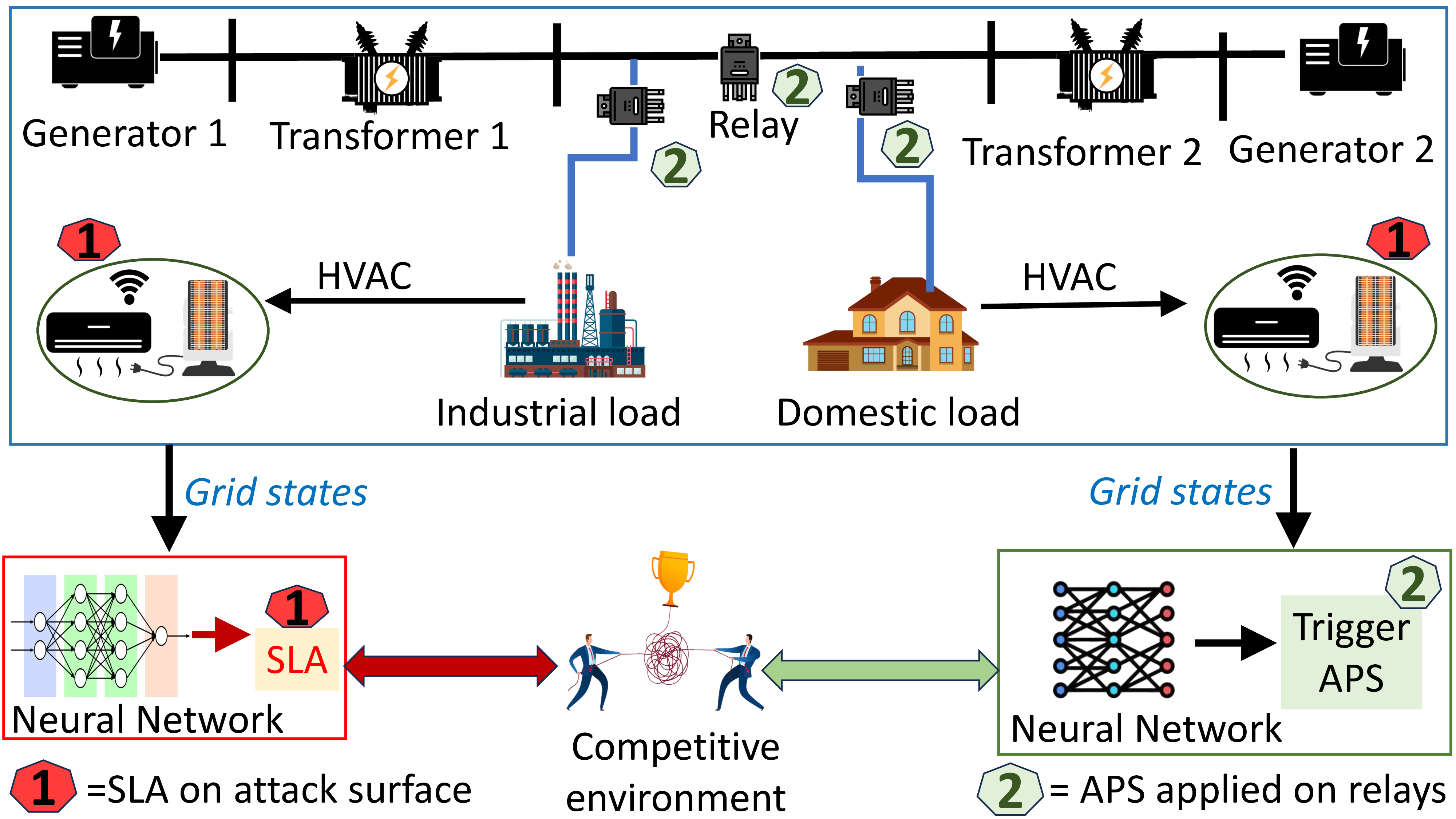} 
    \caption{Operational Framework of the Proposed Adaptive Protection System (APS)}
    \label{fig:system_model}
    \vspace{-3mm}
\end{figure}
%
\section{Preliminaries}
\label{sec:ssystem model}

In this section, we discuss the fundamental operation of smart grids and the related theoretical concepts.

\textbf{Grid Operation:}
\label{subsec:GridOperation}
Smart grids are functionally  divided into three primary subsystems. $(i)$ The `generation' subsystem consists of generators that convert mechanical input energy (e.g., steam, wind, solar) into electrical energy. Each generator unit is connected to an Automatic Generation Control (AGC) unit, which adjusts electrical power production based on the actual power demand in the grid \cite{kundur2007power}. $(ii)$ The `transmission' subsystem transports electrical power produced by the generators to substations via transmission lines. 
($iii$) The `distribution' subsystem   is responsible for the distribution of electricity from substations to end consumers, both domestic and industrial. 
\par In a grid, 
the power flow in electrical `buses' connecting various components can be described as follows. Let $V_i$ and $V_j$ represent the voltages at buses $i$ and $j$, respectively. Let $G_{ij}$ and $B_{ij}$ denote the conductance and susceptance of the transmission line connecting bus $i$ and bus $j$. Additionally, let $\theta_i$ and $\theta_j$ represent the voltage angles at buses $i$ and $j$, respectively. The active ($P$) and reactive ($Q$) power flow from the $i$-th bus to the $j$-th bus in an $n$-bus grid can be described by the power flow equations given below.
%
\begin{align}
\label{eqSys}\nonumber 
    P_{ij} &= |V_i| \sum_{j=1}^{n} |V_j| \left( G_{ij} \cos(\theta_i - \theta_j) + B_{ij} \sin(\theta_i - \theta_j) \right) \\
    Q_{ij} &= |V_i| \sum_{j=1}^{n} |V_j| \left( G_{ij} \sin(\theta_i - \theta_j) - B_{ij} \cos(\theta_i - \theta_j) \right)
\end{align}
%
The stability of each bus, particularly in the presence of disturbances like voltage fluctuations caused by transmission line faults, is referred to as \emph{Fast Voltage Stability Index} (FVSI) \cite{kundur2007power}. It relates the reactive power demand ($Q$) at a given bus to the voltage ($V$) at that bus and indicates how close the grid is to voltage instability or blackout \cite{soltan2018blackiot}. FVSI  is mathematically represented by the following equation for a bus $i$ connected to bus $j \leq n$ via a transmission line in an $n$-bus grid.
\begin{equation} 
\label{eqFVSI}
FVSI_{ij} = \frac{4Z^2_{ij}Q_i(B^2_{ij}+G^2_{ij})}{|V^2_i|B_{ij}} 
\end{equation}
$FVSI_{i,j}>1$ for any transmission line connecting bus $i$ with bus $j (\neq i)$ indicates voltage instability at $i$-th bus. An FVSI value less than 1 depicts voltage stability of the same. Grid operators use FVSI for real time stability assessment. Instability is handled by protection schemes as discussed next. 

\textbf{Grid Protection Systems}
\label{subsec:protection}
During events such as short circuit faults in transmission lines, generator outages, or an increase in power demand, the bus voltages and grid operating frequency may become distorted and, in severe cases, fall outside the safe operational limits of [176, 286] kilovolts (kV) for voltage and [59.5, 60.5] Hz for frequency \cite{kundur2007power}. To address these situations, grid operators deploy protection schemes to maintain safe grid operation. These schemes aim to swiftly isolate the affected area by disconnecting its loads, generators, and transmission lines to minimise damage to grid components. In our work, we consider the following protection schemes, which are widely used in bulk power systems \cite{maiti2023targeted}. \textit{(i) Under Voltage (UV) Protection:} When the voltage at a particular $j$-th bus ($V_j$) falls below the lower threshold of $V^l_{Th_j} = 176$ kV, the UV protection system is triggered. This system disconnects the transmission lines connected to the $j$-th bus after an activation delay of $t_{uv}$ minutes, where $t_{uv} = \frac{0.5}{1 - \left(\frac{V_j}{V^l_{Th_j}}\right)}$. \textit{(ii) Over Voltage (OV) Protection:} When the voltage at a particular $j$-th bus ($V_j$) exceeds the upper threshold of $V^u_{Th_j} = 286$ kV, the OV protection system is triggered. This system disconnects the transmission lines connected to the $j$-th bus after an activation delay of $t_{ov}$ minutes, where $t_{ov} = \frac{0.5}{\left(\frac{V_j}{V^u_{Th_j}}\right) - 1}$. \textit{(iii) Load Shedding:} When the operating frequency of a grid bus falls below the safety threshold of 59.5 Hz, the \emph{Under Frequency Load Shedding} (UFLS) protection system is activated. This system trips the load connected to the bus after a 9-minute delay \cite{huang2019not}. Similarly, when the bus frequency exceeds the upper safety threshold of 60.5 Hz, the \emph{Over Frequency Load Shedding} (OFLS) protection scheme is triggered, tripping the load at the bus with the same 9-minute delay \cite{huang2019not}. By tripping loads and disconnecting transmission lines, these schemes ensure reliable grid operation during fault conditions.

\par
\textbf{Load Alteration Attack Model}
\label{subsecLaa}
A \emph{Load Alteration Attacker} (LAA) in the grid~\cite{shekari2022madiot, soltan2018blackiot} aims to manipulate the active power consumption by varying the loads (both domestic and industrial) connected to the grid by injecting an additional power demand $\Delta P_i$ to bus $i$. Let $P_{ij}$ represent the power flow from the $i$-th bus to the $j$-th bus.
Therefore the power flow from the $i$-th bus to the $j$-th bus due to the LAA ($\tilde{P_{ij}}$), becomes $\tilde{P_{ij}} = P_{ij} + \Delta P_i$. Due to the distributed nature of smart grids, this LAA affects the voltage magnitudes and phase angles at the $i$-th and $j$-th buses. Therefore, the power flow from the $i$-th bus to the $j$-th changes as per the equation given below. 
\begin{align}
\label{eqprem1}
     \tilde{P}_{ij} = |\tilde{V_i}| \sum_{j=1}^{n} |\tilde{V_j}| \left( G_{ij} \cos(\tilde{\theta_i} - \tilde{\theta_j}) + B_{ij} \sin(\tilde{\theta_i} - \tilde{\theta_j}) \right)
\end{align}
Here, $\tilde{\theta_i}$ and $\tilde{\theta_j}$ denote the voltage phase angles at buses $i$ and $j$, respectively, and $\tilde{V_i}$ and $\tilde{V_j}$ represent the voltage magnitudes at buses $i$ and $j$ due to the LAA. The LAA models proposed by existing research work \cite{shekari2022madiot, soltan2018blackiot} trigger the protection schemes
\cite{huang2019not}. In this work, we develop a two-stage load alteration attacker that avoids triggering the existing protection schemes. 
\section{Proposed Methodology}
\label{secMethod}
The overview of our proposed methodology is presented in Fig.~\ref{fig:methodology}. The proposed framework, as depicted, takes the following inputs: {\bf (i)} a smart grid model, {\bf (ii)} its protection schemes and {\bf (iii)} HVAC loads as a set of potential attack surfaces. The framework has two primary components: {\bf (i)} smart grid vulnerability analysis and LAA modelling; to design and inject stealthy load alterations at optimal attack surfaces aiming for blackout without triggering protection systems, and {\bf (ii)} adaptive protection system (APS) modelling: for updating thresholds of the existing voltage protection systems for timely activation of protection schemes to avoid blackouts (see Fig.~\ref{fig:methodology}). 
Placing these components in a competitive environment, the proposed framework formulates a {\em two-player zero-sum Markov game} to theoretically ensure the existence of a learnable stationary policy/strategy that mitigates the effect of any LAA. Finally, the learned APS policy is deployed to counteract the effects of potential SLAs (refer Fig. \ref{fig:methodology}). In the following sections, we discuss the design of each component and the working of each step of the proposed framework in detail.
\begin{figure}[h!]
    \centering
    \includegraphics[width=0.92\linewidth]{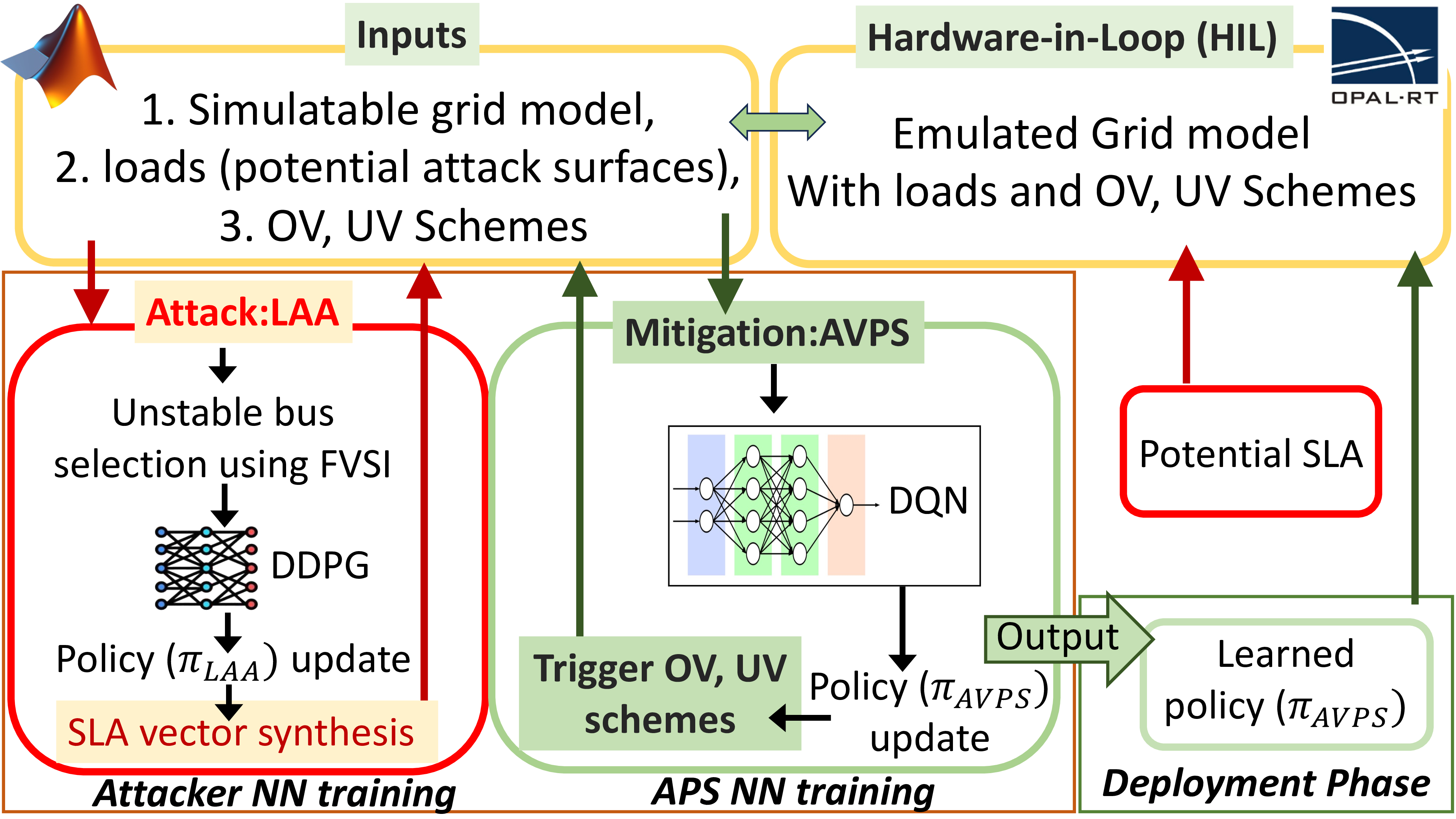} 
    \caption{Our proposed attack mitigation framework}
    \vspace{-2mm}
    \label{fig:methodology}
\end{figure}
\subsection{Novel Attacker-Mitigator Game in Smart Grid} 
\label{subsecGame}
We define a Markov game 
for  grid protection given as,
$\mathbf{\mathcal{G}_{prot}=\langle \mathcal{S}, \{avps,laa\}, A_{laa}\times A_{avps}, \{r_{laa}, r_{avps}\}, \mathcal{P}^{\prime},  \mu_{0}\rangle}$.
Here $\mathcal{S}\subset \mathbb{R}^{n^{\prime}}$ denotes a set of states ($n'$ is the dimension of state-space) consisting of the {\em (i)} active and reactive power flow $P_{ij}$ and $Q_{ij}$ through the transmission line connecting every $i$-th and $j$-th bus, {\em (ii)} voltage magnitude $V_i$ and angle $\theta_i$ of each $i$-th bus in the $n$-bus smart grid along with {\em (iii)} their voltage and frequency-based protection system thresholds $V^{l}_{Th_i}$ and $f_{Th}$ respectively. We consider the following two players in this game, a load alteration attacker $laa$ and an attack mitigator player $avps$. $A_i$, $r_i$ denotes the compact and non-empty action space and pay-off function for a player $i$ such that $r_i: \mathcal{S}\times A_i \mapsto \mathbb{R}$. The Markov game $\mathcal{G}_{prot}$ starts with some initial state $s^{(0)}\sim\mu_{0}$ (also, $s^{(0)}\in \mathcal{S}$), where $\mu_0$ denotes a Probability Density Function (PDF) over state space $\mathcal{S}$. Player 1, $laa$ and Player 2, $avps$, at every discrete time step $t\in \mathbb{N}^{+}$ encounters state $s^{(t)}\sim \mathcal{P}(.|s^{(t-1)},\langle a_{laa}^{(t-1)}, a_{avps}^{(t-1)}\rangle)$ and take actions $a_{laa}^{(t)}\in A_{laa}$ and $a_{avps}^{(t)}\in A_{avps}$ respectively. Here $\mathcal{P}$ is a Gaussian conditional PDF over state space $\mathcal{S}$ ($s^{(t)}\in \mathcal{S}$) with mean $\mu\in\mathcal{S}$ and measurement noise variance $R\in \mathbb{R}^{n^{\prime} \times n^{\prime}}$. Naturally, $\mu$ evolves following Eq.~\ref{eqSys} during usual smart grid operations. The actions of each player at every step are chosen following a policy function. For $laa$, the policy function is defined as $\pi_{laa}:\mathcal{S}\mapsto A_{laa}$ and For $avps$, it is $\pi_{laa}:\mathcal{S}\mapsto A_{laa}$.
%
The player $laa$ is a rogue consumer who manipulates its power consumption in such a way that it is not flagged by frequency or voltage protection systems, but still leads the grid towards a blackout condition. These power manipulations are known as stealthy load alterations (SLAs).
\par$\blacktriangleright$ {\em \underline{A Stealthy load alteration (SLA)} sequence is a sensor data falsification sequence injected in an IoT device, connected as a load in a grid bus such that \emph{a.} the resultant voltage of the bus deviates from its normal profile, leading to blackout condition;  \emph{b.} the grid parameters, frequency and voltage, remain within the permissible operational limits of the traditional protection systems (refer to section \ref{subsec:protection}), thus avoiding their activation before the blackout occurs.}
For example, consider that the IoT-enabled HVAC loads connected to every bus are input to our framework as the potential attack surface. On a side note, high wattage IoT enabled HVAC loads can be identified by using network scanning tools like Nmap~\cite{lyon2009nmap} by scanning in a specific grid region. Power consumption of these HVAC systems can be manipulated by falsifying the operating temperatures ($T$) sensed by such an IoT device as part of the HVAC. The designed SLA in this case is, therefore, a sequence of falsified temperature data $\{\Delta T^{(1)},\cdots,\Delta T^{(L)}\}$, where $\Delta T^{(t)} = \tilde{T}^{(t)} - T^{(t)} \forall t\in [1,L]$. $T^{(t)}$ and $\tilde{T}^{(t)}$ denote the operating temperatures of the HVAC without and with SLA, respectively. 
The change in power consumption $\Delta P_i = \tilde{P}_{ij}-P_{ij}$ is proportional to the temperature change $\Delta T$ as power consumption of the HVACs depends on the heat flow \cite{HVAC_access}, i.e., $\Delta P_i=k\times \Delta T_i$. Therefore, the attacker player $laa$ observes current state $s^{(t)}$ and induces SLA, i.e., $\Delta T^{(t)}\in A_{laa}$ as action, at every time step $t$ in the attacker-mitigator game $\mathcal{G}_{prot}$. As explained in Sec.~\ref{subsecLaa} under SLA the mean of the PDF $\mathcal{P}^\prime$ is $\mu^\prime$, and it evolves following Eq.~\ref{eqprem1}.
%
\par 
On the other hand, $avps$ represents voltage protection systems that intend to mitigate the load alterations that may lead to blackout conditions by tuning the activation thresholds of existing protection schemes (see Sec. \ref{subsec:protection}). The lower and upper voltage thresholds of these protection schemes are denoted with $V_{Th_i}^{l}$ and $V_{Th_i}^{u}$ i.e., $V_{Th_i}^{l}, V_{Th_i}^{u} \in A_{avps}\subset \mathbb{R}^{+}$.The pay-off functions for both players are designed using a function $f: \mathcal{S}\times A_{laa}\times A_{avps}\mapsto \mathbb{R}$.
\begin{align}
\label{eq:payoff}
\nonumber
    f = & 
    \underbrace{\log\left(1 + k\Delta{T}_{i}\right)}_{\textit{1. Increase power consumption}}-
\underbrace{\log\left(1+\frac{dV_i/dt}{R_{Th_i}}\right)}_{\textit{2. Restrict voltage change}} \\
    & + \underbrace{c_1 \log\left(1 + e^{\frac{-V_i}{V^l_{Th_i}}}\right) + c_2 \log\left(1 + e^{\frac{-V^u_{Th_i}}{V_i}}\right)}_{\textit{3. Avoid UV, OV triggering}}
\end{align}
Here, $R_{Th_i}$ is a threshold value chosen by observing the maximum voltage deviation under noise (i.e., without any fault or attack, under normal grid operation). $c_1$ is a boolean variable that returns a value of zero if $V_i < V^l_{Th_i}$ and one otherwise, and another boolean variable $c_2$ returns a value of zero if $V_i > V^u_{Th_i}$, and one otherwise. We consider pay-off for player $laa$, $r_{laa} = f$ and pay-off for player $avps$, $r_{avps} = -f$. The objective of the players would be to maximise their pay-offs by sequentially choosing proper actions at each discrete time step. {\em Part 1} of $f$ motivates maximisation of the total power deviation over the horizon by inducing SLA $\Delta T_i$ at a bus $i$ at every time step. {\em Part 2} penalises any action of $laa$ that causes rapid deviation in bus voltage $V_i$ beyond the nominal voltage deviation $R_{Th_i}$. Limiting this \emph{rate of change of voltage} to this nominal limit ensures low false positives by the mitigator. {\em Part 3} of $f$ incentivises $laa$ when its actions keeps the bus voltage $V_i \in [V_{th}^{l}, V_{th}^{u}]$. Therefore, $laa$ chooses an action at every time step that maximises $r_{laa}$ or $f$, leading to a blackout. In the next time step, $avps$ chooses a voltage threshold $V_{th}^{l}$ to maximise $r_{avps}$ or minimise $f$, i.e., the $laa$'s pay-off such that the protection system gets activated before a blackout occurs. Note that, we only consider $V_{th}^{u}$ as $avps$'s action since $V_{th}^{u} = \alpha \times V_{th}^{l}$ for any $\alpha\in (1,2]$. Now we discuss how we design a competitive Multi-Agent Reinforcement Learning (MARL) environment as part of the proposed framework to implement this two-player game. 
\subsection{Competitive MARL Environment Design}
\label{subsecMarl}
To implement this attacker-mitigator game $\mathcal{G}_{prot}$, we use two neural networks (NNs) as attacker player $laa$ and mitigator player $avps$. These attacker and mitigator players are placed in the smart grid that is input to our framework to act as LAA and APS components, respectively. As formalized in $\mathcal{G}_{prot}$, the LAA component $laa$ injects SLAs in attack surfaces input to the framework. The APS component $avps$ adjusts the voltage thresholds of the existing protection schemes that are also input to our framework. The LAA component is trained using \emph{Deep Deterministic Policy Gradient} (DDPG) policy to cause blackouts in the grid by injecting SLAs from its continuous action space $A_{laa}$. Whereas, the APS component is trained using Q-learning to adjust the voltage threshold from the discrete action space $A_{avps}$ and counteract the blackouts intended by the LAA component. This is because the threshold decisions are positive values that lie within a convex and compact interval between $80\%$ to $130\%$ of the nominal bus voltage \cite{huang2019not}. Our framework trains these two components in a MARL environment that resembles the attacker-mitigator game. We present the training process using Algo.~\ref{alg:game}. 
\begin{algorithm}
\caption{MARL setup: LAA vs. APS}
\label{alg:game}
\begin{algorithmic}[1]
\State \textbf{Input:} Grid Environment with attack surfaces and protection systems, LAA initial policy $\pi'_{laa}$, APS initial policy $\pi_{avps}$, Episode Length $L$
\label{alg:input}
\State $s^{(0)} \sim \mu_0$ at episode start\Comment{state Init from a distribution}\label{alg:init1}
\For{all $t\in [0, L]$} \label{alg:for}
\State $\vartriangleright$ \textbf{LAA Player 1:} \label{alg:attacker_start}
\State $si \gets 1$, $s^{(t)} \gets \{\}$, $'i, j' \gets 0$ \Comment{episode init} \label{alg:init}
\For{$i = 1$ \textbf{to} $n$} \Comment{for an $n$-bus grid} \label{alg:attacker_for_loop_start}
    \For{$j = 1$ \textbf{to} $n$} \label{alg:attacket_fori}
        \If{$FVSI_{i,j} > si$} \Comment{FVSI comparison} \label{alg:attacker_if}
            \State $si \gets FVSI_{i,j}$, $ui \gets i$, $uj \gets j$ 
            \label{alg:fvsi_update}
        \EndIf \label{alg:attacker_if_end}
    \EndFor \label{alg:attacker_for_loop_end_1}
    \State $s^{(t)} \gets s^{(t)} \cup \{ P_{i'j'}, V_{i'}, \theta_{i'}, V^l_{Th_{i'}}\}$ 
    \label{alg:attacker_observe}
    \State $a^{(t)}_{laa} = \pi'_{laa}(s^{(t)})$ \Comment{SLAs on Most Unstable bus} \label{alg:attacker_action}
    \State Compute ${r'}_{laa}^{(t)}$ \Comment{Compute LAA reward} \label{alg:attacker_reward}
\EndFor \label{alg:attacker_for_loop_end}
\State $\vartriangleright$\textbf{ APS Player 2:} \label{alg:aps_start}
\State $a^{(t)}_{avps} = \pi_{avps}(s^{(t)})$ \Comment{Adaptive Threshold Update} \label{alg:aps_action}
\State Compute ${r}_{avps}^{(t)}$ \Comment{Compute APS reward} \label{alg:aps_reward}
\EndFor \label{alg:episode_for_loop_end}
\State Update $\pi'_{laa}, \pi_{avps}$ at every Epoch-End\Comment{Policy Updates} \label{alg:attacker_update}
\label{alg:aps_update}
\State \textbf{Output:} Optimal Mitigator policy  $\pi^*_{avps}$ \label{alg:output}
\end{algorithmic}
\end{algorithm}
\par Algo.~\ref{alg:game} takes as input the smart grid model consisting of vulnerable loads and protection systems, the initial policies of the LAA and APS agents, and the episode length $L$ (in terms of time steps) for training the agents (line~\ref{alg:input}). At the start of every episode, the grid states are initialised from the $\mu_0$ distribution as considered for the Markov game $\mathcal{G}_{prot}$(line~\ref{alg:init1}). Algo.~\ref{alg:game} then iterates over an entire episode length to let the agents interact with the competitive environment (line \ref{alg:for}). The training of the LAA agent (Player 1 $laa$) begins by initialising the variable $si$ to 1. $si$ stores the highest FVSI. An empty set $s^{(t)}$ to store the state information at time step $t$ as the participating agents' observations, and the variables $i',j'$ to store the bus numbers with the highest FVSI (lines \ref{alg:attacker_start}-\ref{alg:init}). 
\par In {\em step 1} of the attack, the LAA DRL agent computes $FVSI$ of each transmission line following Eq.~\ref{eqFVSI} for a varying load profile. To integrate model-specific knowledge into the attack process, the intelligent LAA agent selects the \emph{most unstable bus} connected with a transmission line, having the highest FVSI. This makes the proposed LAA model outperform the existing ones in the literature~\cite{shekari2022madiot, soltan2018blackiot} by ensuring stealth. To achieve this, for each bus pair ($i,j$), Algo. \ref{alg:game} compares the FVSI values. If the FVSI value is greater than the current value of $si$, Algo. \ref{alg:game} updates $si$ with the new FVSI value and sets $i'$ to $i$, and $j'$ to $j$, indicating that bus $i$ is the most unstable at current time step (lines \ref{alg:attacket_fori}-\ref{alg:fvsi_update}). Algo. \ref{alg:game} then stores the observation variables for this most unstable bus in the grid in the set $s^{(t)}$ (line \ref{alg:attacker_observe}). In {\em step 2} of the attack, the LAA agent takes an action by inducing SLAs at the most unstable bus $i'$, based on these observed states (line \ref{alg:attacker_action}), followed by the computation of its reward (line \ref{alg:attacker_reward}). Basically, at each time step $t$ LAA agent observes the states of the grid $s^{(t)}\in \mathcal{S}$ and generates SLA values as actions $\Delta T_{i'}=a^{(t)}_{laa}\in A_{laa}$ following a policy $\pi_{laa}$ such that its reward $r_{laa}$ is maximised.
\par Following this, Algo. \ref{alg:game} starts the operation of the APS agent (Player 2 $avps$) (line \ref{alg:aps_start}). The APS agent adjusts the activation thresholds of the existing OV and UV protection schemes based on the observed states in $s^{(t)}$ (line \ref{alg:aps_action}). The reward for the APS agent is then computed (line \ref{alg:aps_reward}). The APS DRL agent, at each time-step $t$, observes the same grid parameters as $avps$ and outputs the adaptive thresholds for existing voltage protection systems $V_{Th}^{l}=a_{avps}^{(t)}\in A_{avps}$ as actions following a policy $\pi_{avps}$ such that its reward $r_{avps}$ is maximised. At the end of every epoch, both agents' policies are updated (line \ref{alg:attacker_update}). The critic network in the LAA DDPG agent updates its actor network's policy $\pi_{laa}$ (responsible for generating SLAs) such that the expected value of an action-value function $Q$ is maximised. The DQN, used as the APS agent, also updates its policy $\pi_{avps}$ by maximising a similar $Q$ function in its discrete action space. In line \ref{alg:output}, Algo.~\ref{alg:game} outputs the optimal mitigation policy for the APS agent.
\subsection{Correctness of the Algorithm}
\label{subsecGameThProof}
The sequential action of the players in the competitive MARL environment, with opposing objectives, makes $\mathcal{G}_{prot}$ a two-player \emph{min-max Stackelberg game}. Therefore looking for a Stackelberg equilibrium  in this game is a search for a joint policy $\pi^*_{laa}, \pi^*_{avps}$ for players $laa$ and $avps$ respectively that solves the following optimisation problem.
\begin{align}
\label{eqOpt}
\nonumber
 \pi^*_{laa}, \pi^*_{avps} = &\underset{\pi_{avps}:\mathcal{S}\mapsto A_{avps}}{\text{arg min}} \quad
\underset{\pi_{laa}:\mathcal{S}\mapsto A_{laa}}{\text{arg max}} \quad 
 \mathbb{E}[ \sum^{L}_{i=1} \gamma^t f] \\
 = &\underset{\pi_{avps}}{\text{arg min}}\quad 
\underset{\pi_{laa}}{\text{arg max}} \quad J_{\pi_{laa}, \pi_{avps}}
\end{align}
Here $L\in \mathbb{N}^+$ is a large horizon for the discrete game,  $\gamma \in (0,1]$ is a discount factor for the future pay-off values and $J_{\pi_{laa}, \pi_{avps} }$ is the objective of the optimisation problem for policies $\pi_{laa}, \pi_{avps}$ respectively for the attacker player $laa$ and mitigator player $avps$. 
\par\noindent $\bullet$ \textbf{Completeness:} Notice that, $f$ is designed in such a way that {\em (i)} it is continuous and smooth in the action spaces $A_{laa}$ and $A_{avps}$; 
{\em (ii)} second derivative of $f$ w.r.t. $\Delta T_{i}$ is negative in $A_{laa}$, and {\em (iii)} second derivative of $f$ w.r.t. $V_{Th_i}^{u}$ is positive in $A_{avps}$. 
Therefore $f$ is \textbf{convex} w.r.t. $V_{Th_i}^{u}$ and \textbf{concave} w.r.t. $\Delta T_{i}$. 
As we know from \emph{Jensen's inequalities}, convexity (and concavity) holds even in the presence of an expectation operator~\cite{zapala2000jensen}. Therefore, the fact that $f$ is convex in action space $A_{avps}$ and concave in action space $A_{laa}$ implies the objective function $J$ in Eq.~\ref{eqOpt} is also convex in $A_{avps}$ and concave in action space $A_{laa}$.
This makes the zero-sum min-max Stackelberg game $\mathcal{G}_{prot}$ a {\emph convex-concave Stackelberg game}. Using the findings of~\cite{goktas2022zero}, we can claim the following for the attacker-mitigator game $\mathcal{G}_{prot}$. 
\begin{Claim}
\label{claim1}
In an attacker-mitigator game, modelled like $\mathcal{G}_{prot}$ there always exists a stationary policy $\mathbf{\pi^{*}_{avps}}$ for player 2 $avps$ that can mitigate any actions of player 1 $laa$ having any policy.
\end{Claim}
\begin{proof}
As shown in~\cite{goktas2022zero}, for any zero-sum convex-concave  Stackelberg Game like $\mathcal{G}_{prot}$, there always exists a Stackelberg equilibrium for a policy pair $\langle \pi^{*}_{laa}, \pi^*_{avps}\rangle$ that are the stationary points for the game as they are the solution of the optimisation problem Eq.~\ref{eqOpt}. Therefore, for any suboptimal policy $\pi_{laa}$ of attacker player $laa$, the objective function satisfies $J_{\pi_{laa},\pi^*_{avps}} \geq J_{\pi^*_{laa},\pi^*_{avps}}$. This proves the fact that there always exists a stationary mitigator's policy $\pi^*_{avps}$ that can mitigate actions taken by any suboptimal attacker policy $\pi_{laa}$ as it can mitigate the optimal one $\pi^*_{laa}$.
\end{proof}
%
%
%
\begin{figure*}[!ht]
    \begin{subfigure}{0.65\columnwidth}
\includegraphics[width=\linewidth]{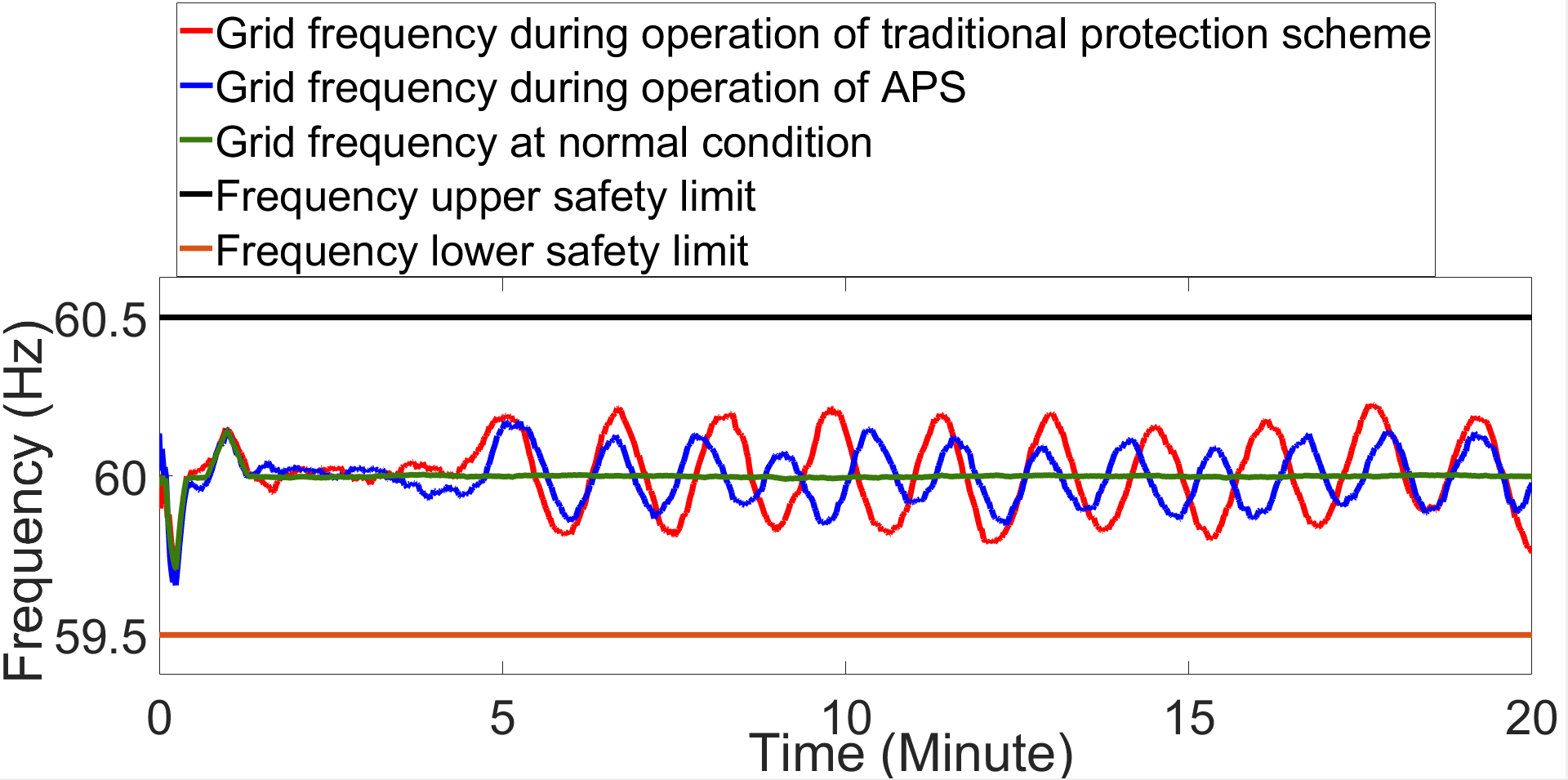}
  \caption{Grid operating frequency under the influence of LAA.}
  \label{freq14}
    \end{subfigure}
    \hfill
    \begin{subfigure}{0.65\columnwidth}
\includegraphics[width=1\textwidth]{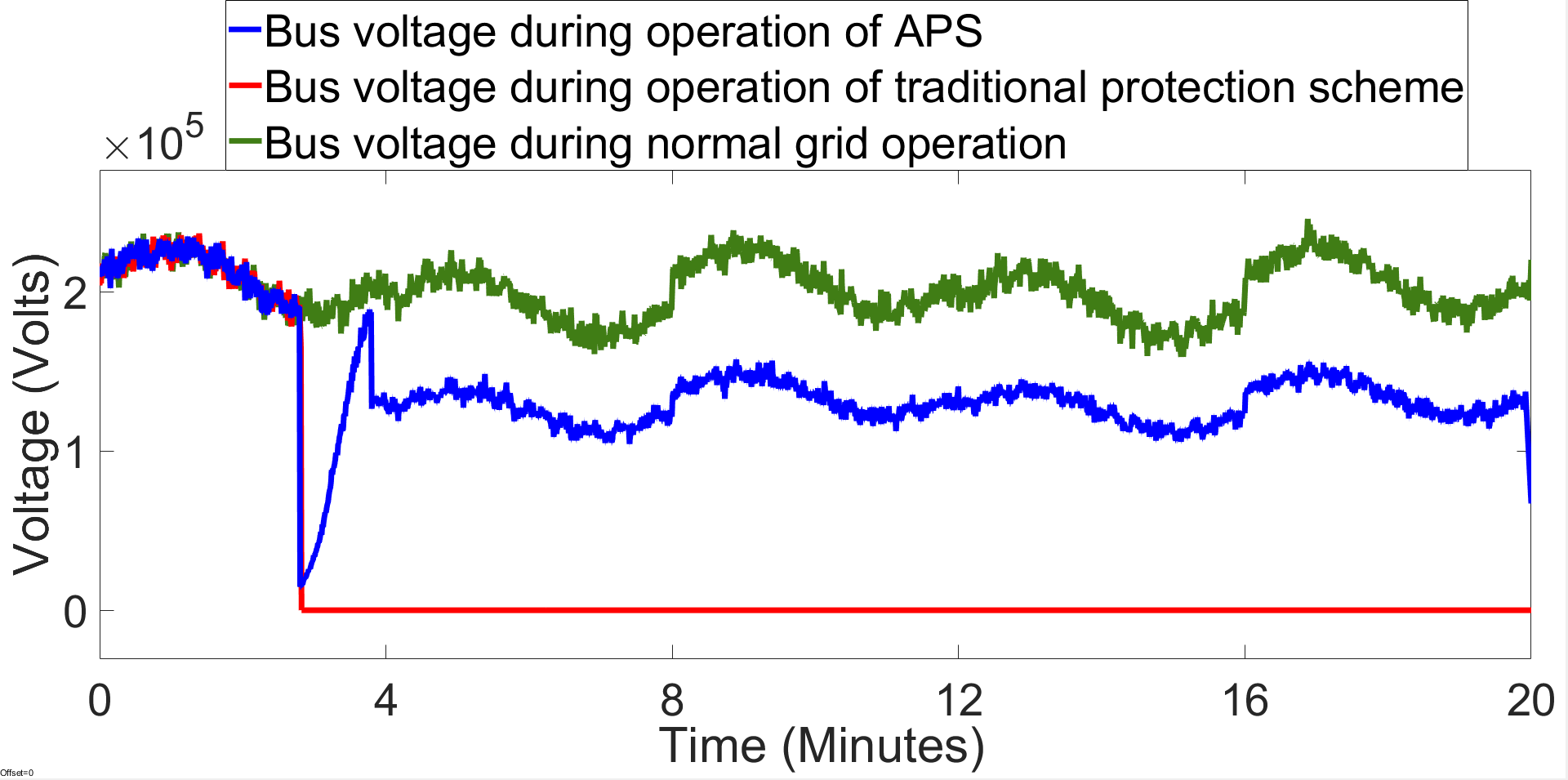}
  \caption{Voltage waveform of bus 3 (phase a) during LAA.}
  \label{volt14}
    \end{subfigure}
    \hfill
    \begin{subfigure}{0.65\columnwidth}
        \includegraphics[width=1\textwidth]{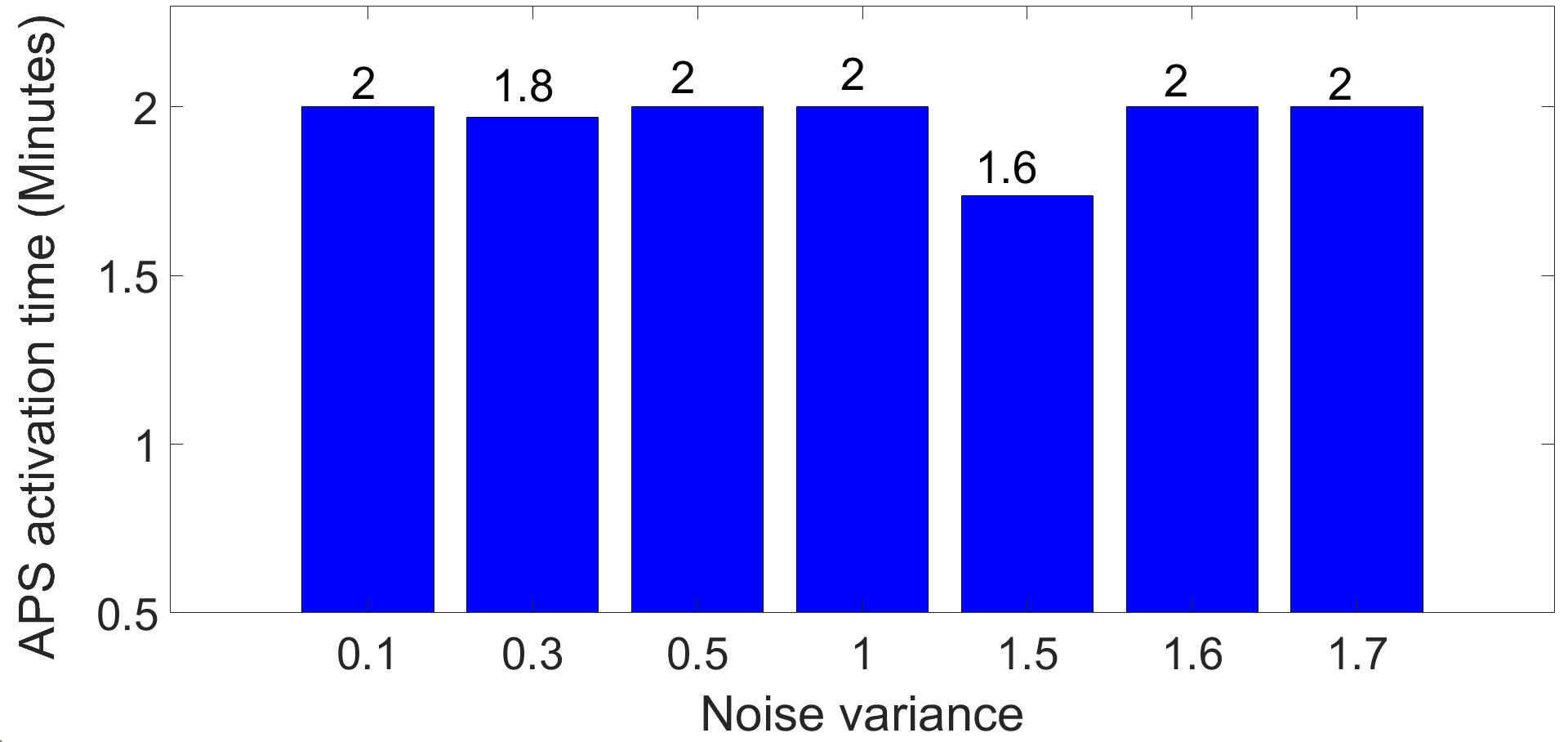}
  \caption{Activation time of the APS in the presence of noise in the HVAC load profile.}
  \label{FAR14}
    \end{subfigure}
 
 \caption{Real-time simulation results of the IEEE 14 bus model}
 \vspace{-5mm}
 \label{figAttackSimulationHIL_3}
\end{figure*}
\par\noindent$\bullet$ \textbf{Soundness:} 
%
The action-value function $Q$ is maximised to update policies in both DDPG and DQN. For an $L$ length sequence of actions, the $Q$ value resembles the cost $J$ from Eq.~\ref{eqOpt}. Hence, the competitive MARL in Algo.~\ref{alg:game} also resembles the convex-concave game $\mathcal{G}_{prot}$ and thus learns a stationary policy that solves Eq.~\ref{eqOpt} as claimed in~\cite{goktas2022zero}. Therefore the policy learnt at the end of the training converges to a sound policy $\pi^*_{avps}$ that mitigates any SLA following Claim~\ref{claim1}.
\section{Experimental Results}
\label{secExp}
We present our experimental results for the proposed framework using a standard IEEE 14-bus power grid model, designed in Matlab Simulink. It is equipped with dynamic HVAC loads as potential attack surfaces, and existing protection systems, as discussed in Section \ref{sec:ssystem model}. The HVAC loads, with variable power demand profiles, were designed using Matlab's standard house heating system model. Both the attacker and APS agents were developed using the DRL toolbox in Matlab, running on a workstation with a 16-core Intel Xeon CPU, 64 GB of RAM, and 16 GB NVIDIA Quadro GPU. The training process, as described in Algo. \ref{alg:game}, was conducted over 500 episodes, each spanning 120,000 time steps with a sampling time of 0.01 s, resulting in a total training duration of 20 minutes per episode. The codebase for our framework is available in~\cite{GIT}.
The learnt APS policy is emulated in real-time along with the learnt LAA policy using Opal RT OP4510 HIL setup with Artemis solver. The plots in Fig.~\ref{figAttackSimulationHIL_3} demonstrate the effectiveness of the output APS policy
by plotting the following grid parameters.
\par\noindent\textbf{$(i)$ Grid operating frequency:} Fig.~\ref{freq14} illustrates the grid's operating frequency waveform under SLA that increases the power demand of the HVAC loads attached to bus 3 (with the highest FVSI). This is achieved by manipulating the HVAC operating temperature during its peak power consumption time, which occurs at 2 minutes into the real-time simulation. The peak power consumption period was identified by analyzing the load profile of Matlab's house heating system model. As shown in Fig. \ref{freq14}, the grid's operating frequency (green plot) stabilizes to the nominal 60 Hz range after 3 minutes during normal operation, without the presence of SLA. Under SLA, when the traditional \emph{Over Voltage} (OV) and \emph{Under Voltage} (UV) protection schemes are active, the operating frequency (red plot) oscillates around the 60 Hz mark but remains within the upper (black plot) and lower (red plot) safety limits. Similarly, under SLA and in the presence of an APS agent, the frequency (blue plot) stays within the safety limits. This occurs because AGC interprets the increase in power demand as normal during peak consumption and adaptively adjusts the power production to meet the increased demand, keeping the frequency trajectory within the safe operational limits. As a result, the LAA remains undetected by the grid operators.
\par \noindent\textbf{$(ii)$ Voltage waveform at bus 3:} In Fig.~\ref{volt14}, we plot the voltage waveform of bus 3 (phase $a$) under SLA. As can be seen, the voltage waveform (green plot) during normal grid operation, without the SLA, follows a sinusoidal pattern and operates at its rated value of $2 \times 10^5$ volts. However, under SLA, at the 2-minute mark, the voltage waveform (red plot) immediately drops to zero, indicating a {\em blackout condition} caused by the SLA. This occurs because the traditional OV and UV protection schemes are activated with a predefined delay after the grid voltage trajectory crosses their fixed activation thresholds. In contrast, our APS scheme dynamically adjusts the activation thresholds of the OV and UV schemes, triggering them at the most appropriate time to mitigate the SLA. This allows the APS agent to maintain the voltage at bus 3 at an operational level of $1 \times 10^5$ volts (blue plot), thereby averting the blackout.
\par\noindent\textbf{$(iii)$ False Positive Rates (FPR) of the APS Policy:} The operation of the APS agent depends on the power demand or load profile of the grid, which constitutes its observation space. Therefore, it is crucial to ensure the reliable operation of the APS agent in real-world grid settings, even in the presence of errors in the observation space, such as noise in the load profile. To evaluate the APS agent's accuracy in triggering the OV and UV schemes under noisy conditions, we first introduced uniformly distributed noise with variance values of (0.1, 0.3, 0.5, 1, 1.5, 1.6, 1.7) MW into the HVAC load profiles of the IEEE 14-bus model. For each noise value, we ran 10 real-time simulations of the SLA on the grid, initiating the attack at the 2-minute mark. We then plotted the triggering times of the APS agent for each noise scenario across all runs. With this experiment we evaluate how the incorporation of nominal RoCoV as a threshold parameter reduces the FPR of the proposed APS.
\par As shown by the blue bar plots in Fig.~\ref{FAR14}, for noise variance values of 0.1, 0.5, 1, 1.6, and 1.7 MW, the APS agent successfully triggered the OV and UV schemes at the appropriate time (2 minutes) in all runs, mitigating the attack. However, in the 0.3 and 1.5 MW noise scenarios, the APS was triggered prematurely. This results in a FPR of 2/7, which is significantly low. Thus, our APS agent demonstrates effective mitigation of LAAs, even in the presence of noise in its observation space (load profile). 
%
%
\section{Conclusion and Future Work}
\label{secConcl}
This work proposes an adaptive policy for voltage-based protection systems in a smart grid to mitigate load alteration attack sequences that stealthily attempt to cause blackouts in the grid. For a given grid model the proposed framework outputs an optimal policy for adaptive voltage protection that is game-theoretically proven to mitigate any stealthy load alterations by timely activating the voltage protection systems in the grid. In the future, we would like to extend our framework to design adaptive protection policies for other relevant classes of protection systems and also test the current method on larger models (e.g. IEEE 118 Bus).
\bibliographystyle{IEEEtran}
\bibliography{IEEEabrv,ref2}

\begin{thebibliography}{10}
\providecommand{\url}[1]{#1}
\csname url@rmstyle\endcsname
\providecommand{\newblock}{\relax}
\providecommand{\bibinfo}[2]{#2}
\providecommand\BIBentrySTDinterwordspacing{\spaceskip=0pt\relax}
\providecommand\BIBentryALTinterwordstretchfactor{4}
\providecommand\BIBentryALTinterwordspacing{\spaceskip=\fontdimen2\font plus
\BIBentryALTinterwordstretchfactor\fontdimen3\font minus \fontdimen4\font\relax}
\providecommand\BIBforeignlanguage[2]{{%
\expandafter\ifx\csname l@#1\endcsname\relax
\typeout{** WARNING: IEEEtran.bst: No hyphenation pattern has been}%
\typeout{** loaded for the language `#1'. Using the pattern for}%
\typeout{** the default language instead.}%
\else
\language=\csname l@#1\endcsname
\fi
#2}}

\bibitem{soltan2018blackiot}
S.~Soltan, P.~Mittal, and H.~V. Poor, ``$\{$BlackIoT$\}$:$\{$IoT$\}$ botnet of high wattage devices can disrupt the power grid,'' in \emph{27th USENIX Security Symposium (USENIX Security 18)}, 2018, pp. 15--32.

\bibitem{yang2023resilient}
S.~Yang, K.-W. Lao, Y.~Chen, and H.~Hui, ``Resilient distributed control against false data injection attacks for demand response,'' \emph{IEEE Transactions on Power Systems}, vol.~39, no.~2, pp. 2837--2853, 2023.

\bibitem{shekari2022madiot}
T.~Shekari, A.~A. Cardenas, and R.~Beyah, ``$\{$MaDIoT$\}$ 2.0: Modern $\{$High-Wattage$\}$$\{$IoT$\}$ botnet attacks and defenses,'' in \emph{31st USENIX Security Symposium (USENIX Security 22)}, 2022, pp. 3539--3556.

\bibitem{huang2019not}
B.~Huang, A.~A. Cardenas, and R.~Baldick, ``Not everything is dark and gloomy: Power grid protections against $\{$IoT$\}$ demand attacks,'' in \emph{28th USENIX Security Symposium (USENIX Security 19)}, 2019, pp. 1115--1132.

\bibitem{HVAC_access}
C.~Toyama, ``Mitsubishi electric air conditioning system,'' \emph{Cybersecurity and Infrastructure Security Agency}, 2021.

\bibitem{koley2021catch}
I.~Koley, S.~Adhikary, and S.~Dey, ``Catch me if you learn: Real-time attack detection and mitigation in learning enabled cps,'' in \emph{2021 IEEE Real-Time Systems Symposium (RTSS)}.\hskip 1em plus 0.5em minus 0.4em\relax IEEE, 2021, pp. 136--148.

\bibitem{tan2016optimal}
R.~Tan, H.~H. Nguyen, E.~Y. Foo, X.~Dong, D.~K. Yau, Z.~Kalbarczyk, R.~K. Iyer, and H.~B. Gooi, ``Optimal false data injection attack against automatic generation control in power grids,'' in \emph{2016 ACM/IEEE 7th International Conference on Cyber-Physical Systems (ICCPS)}.\hskip 1em plus 0.5em minus 0.4em\relax IEEE, 2016, pp. 1--10.

\bibitem{ni2019multistage}
Z.~Ni and S.~Paul, ``A multistage game in smart grid security: A reinforcement learning solution,'' \emph{IEEE transactions on neural networks and learning systems}, vol.~30, no.~9, pp. 2684--2695, 2019.

\bibitem{cunningham2022deep}
J.~D. Cunningham, A.~Aved, D.~Ferris, P.~Morrone, and C.~S. Tucker, ``A deep learning game theoretic model for defending against large scale smart grid attacks,'' \emph{IEEE Transactions on Smart Grid}, vol.~14, no.~2, pp. 1188--1197, 2022.

\bibitem{umsonst2020nash}
D.~Umsonst, S.~Sarita{\c{s}}, and H.~Sandberg, ``A nash equilibrium-based moving target defense against stealthy sensor attacks,'' in \emph{2020 59th IEEE Conference on Decision and Control (CDC)}.\hskip 1em plus 0.5em minus 0.4em\relax IEEE, 2020, pp. 3772--3778.

\bibitem{kundur2007power}
P.~Kundur, ``Power system stability,'' \emph{Power system stability and control}, vol.~10, pp. 7--1, 2007.

\bibitem{maiti2023targeted}
S.~Maiti, A.~Balabhaskara, S.~Adhikary, I.~Koley, and S.~Dey, ``Targeted attack synthesis for smart grid vulnerability analysis,'' in \emph{Proceedings of the 2023 ACM SIGSAC Conference on Computer and Communications Security}, 2023, pp. 2576--2590.

\bibitem{lyon2009nmap}
G.~F. Lyon, \emph{Nmap network scanning: The official Nmap project guide to network discovery and security scanning}.\hskip 1em plus 0.5em minus 0.4em\relax Insecure, 2009.

\bibitem{zapala2000jensen}
A.~M. Zapa{\l}a, ``Jensen’s inequality for conditional expectations in banach spaces,'' 2000.

\bibitem{goktas2022zero}
D.~Goktas, S.~Zhao, and A.~Greenwald, ``Zero-sum stochastic stackelberg games,'' \emph{Advances in Neural Information Processing Systems}, vol.~35, pp. 11\,658--11\,672, 2022.

\bibitem{GIT}
``{Adaptive Protection System for Smart Grid},'' \url{https://anonymous.4open.science/r/Adaptive-Protection-System-for-Smart-Grids-8247/README.md}.

\end{thebibliography}

\end{document}